\documentclass[oribibl]{llncs}
\usepackage{epsfig}
\usepackage{graphicx}
\usepackage{bbm}
\usepackage{balance}

\newtheorem{prop}[theorem]{Proposition}
\newtheorem{rem}[theorem]{Remark}

\begin{document}

\title{Dynamic 3-sided Planar Range Queries with Expected Doubly Logarithmic Time\thanks{This work is based on a combination of two conference papers that appeared in Proc. 21st International Symposium on Algorithms and Computation (ISAAC), 2010: pages 1-12 (by all authors except third) and 13th International Conference on Database Theory (ICDT), 2010: pages 34-43 (by all authors except first).}}

\author{Gerth St\o lting Brodal\inst{1}, Alexis C. Kaporis\inst{2}, Apostolos N. Papadopoulos\inst{4}, Spyros Sioutas\inst{3}, Konstantinos Tsakalidis\inst{1}, Kostas Tsichlas\inst{4}}
\institute{MADALGO\thanks{Center for Massive Data Algorithmics, a Center of the Danish National Research
Foundation.}, Department of Computer Science,
Aarhus University, Denmark\\
\email{ \{gerth,tsakalid\}@madalgo.au.dk}
\and
Computer Engineering and Informatics Department,
University of Patras, Greece\\
\email{kaporis@ceid.upatras.gr}
\and
Department of Informatics, Ionian University, Corfu, Greece\\ \email{sioutas@ionio.gr}
\and
Department of Informatics, Aristotle University of Thessaloniki, Greece\\ \email{tsichlas@csd.auth.gr}
}

\maketitle

\begin{abstract}
This work studies the problem of 2-dimensional searching for the 3-sided
range query of the form $[a, b]\times (-\infty, c]$ in both main and
external memory, by considering a variety of input distributions.
We present three sets of solutions each of which examines the 3-sided problem in both RAM and I/O model respectively.
The presented data structures are deterministic and the expectation is with respect to the input distribution:

(1) Under continuous $\mu$-random distributions of the $x$ and $y$ coordinates, we present a dynamic linear main memory solution, which answers 3-sided queries in $O(\log n + t)$ worst case time and scales
with $O(\log\log n)$ expected with high probability update time,
where $n$ is the current number of stored points and $t$ is the size of the query output.
We externalize this solution, gaining $O(\log_{B}n +t/B)$
worst case and $O(log_{B}logn)$ amortized expected with high probability I/Os
for query and update operations respectively, where $B$ is the disk block size.

(2)Then, we assume that the inserted points have their $x$-coordinates drawn from a class of {\em smooth} distributions, whereas the $y$-coordinates are arbitrarily distributed.
The points to be deleted are selected uniformly at random among the inserted points.
In this case we present a dynamic linear main memory solution that supports queries in $O(\log \log n + t)$ expected time with high probability and updates in $O(\log \log n)$ expected amortized time, where $n$ is the number of points stored and $t$ is the size of the output of the query. We externalize this solution, gaining $O(\log \log_B n + t/B)$ expected I/Os with high probability for query operations and $O(\log_B \log n)$ expected amortized I/Os for update operations, where $B$ is the disk block size. The space remains linear $O(n/B)$.

(3)Finally, we assume that the $x$-coordinates are continuously drawn from a smooth
distribution and the $y$-coordinates are continuously drawn from a more restricted class of
realistic distributions. In this case and by combining the Modified Priority Search Tree~\cite{SMKLTTV04} with the Priority Search Tree~\cite{MC81}, we present a dynamic linear main memory solution that supports queries in $O(\log \log n + t)$ expected time with high probability and updates in $O(\log \log n)$ expected time with high probability. We externalize this
solution, obtaining a dynamic data structure that answers 3-sided queries in
$O(\log_B \log n + t/B)$ I/Os expected with high probability, and it can be updated
in $O(\log_B \log n)$ I/Os amortized expected with high probability. The space remains linear $O(n/B)$.
\end{abstract}
\section{Introduction}
Recently, a significant effort has been performed towards developing
worst case efficient data structures for range searching in two dimensions \cite{V01}.
In their pioneering work, Kanellakis et al. \cite{KRVV93}, \cite{KRVV96} illustrated that the
problem of indexing in new data models (such as constraint, temporal
and object models), can be reduced to special cases of two-dimensional
indexing. In particular, they identified the {\it 3-sided range searching}
problem to be of major importance.

The 3-sided range query in the 2-dimensional space is defined by a region of the form
$R = [a, b]\times (-\infty, c]$, i.e., an ``open'' rectangular region, and returns all points
contained in $R$. Figure \ref{fig:3-sided} depicts examples of possible 3-sided queries,
defined by the shaded regions. Black dots represent the points comprising the result.
In many applications, only positive coordinates are used and therefore, the region defining the
3-sided query always touches one of the two axes, according to application semantics.

Consider a time evolving database storing measurements collected from a sensor network.
Assume further, that each measurement is modeled as a multi-attribute tuple of the form
$<$$id,a_1,a_2,...,a_d,time$$>$, where $id$ is the sensor identifier that produced the measurement,
$d$ is the total number of attributes, each $a_i$, $1 \leq i \leq d$, denotes the value
of the specific attribute and finally $time$ records the time that this measurement was produced.
These values may relate to measurements regarding temperature, pressure, humidity, and so on.
Therefore, each tuple is considered as a point in $\mathbbm{R}^d$ space.
Let $F$:~$\mathbbm{R}^d \rightarrow \mathbbm{R}$ be a real-valued ranking function that scores each
point based on the values of the attributes. Usually, the scoring function $F$ is monotone and without
loss of generality we assume that the lower the score the ``better'' the measurement (the other case is symmetric).
Popular scoring functions are the aggregates {\sf sum}, {\sf min}, {\sf avg} or other more complex combinations of the attributes.
Consider the query: ``search for all measurements taken between the time instances $t_1$ and $t_2$
such that the score is below $s$''. Notice that this is essentially a 2-dimensional 3-sided query
with $time$ as the $x$ axis and $score$ as the $y$ axis. Such a transformation from a multi-dimensional
space to the 2-dimensional space is common in applications that require a temporal dimension, where each
tuple is marked with a timestamp storing the arrival time \cite{MBP06}.
This query may be expressed in SQL as follows: \\

\noindent
{\sf SELECT $id$, $score$, $time$}  \\
{\sf FROM SENSOR\_DATA} \\
{\sf WHERE $time$$>=$$t_1$ AND $time$$<=$$t_2$ AND $score$$<=$$s$};
~\\

It is evident, that in order to support such queries, both search and update operations (i.e., insertions/deletions)
must be handled efficiently. Search efficiency directly impacts query response time as well as the general system
performance, whereas update efficiency guarantees that incoming data are stored and organized quickly, thus,
preventing delays due to excessive resource consumption. Notice that fast updates will enable the support of
stream-based query processing \cite{BBD+02} (e.g., continuous queries), where data may arrive at high rates and therefore
the underlying data structures must be very efficient regarding insertions/deletions towards supporting arrivals/expirations
of data. There is a plethora of other applications (e.g., multimedia databases, spatio-temporal) that fit to a scenario similar
to the previous one and they can benefit by efficient indexing schemes for 3-sided queries.

\begin{figure}[!t]
\begin{center}
\label{fig:3-sided}
\includegraphics[scale=0.80]{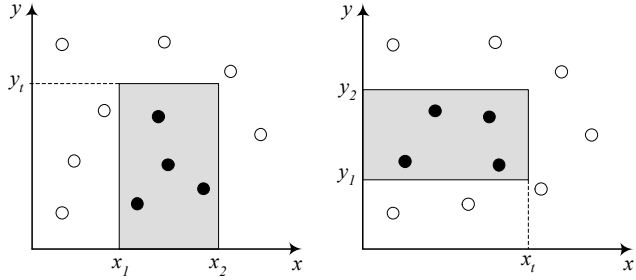}
\end{center}
\caption{Examples of 3-sided queries.}
\end{figure}

Another important issue in such data intensive applications is memory consumption. Evidently, the best practice is to keep
data in main memory if this is possible. However, secondary memory solutions must also be available to cope with
large data volumes. For this reason, in this work we study both cases offering efficient solutions both in the RAM and I/O
computation models. In particular, the rest of the paper is organized as follows. In Section \ref{prel}, we
discuss preliminary concepts, define formally the classes of used probability distributions and present the data structures
that constitute the building blocks of our constructions. Among them, we introduce the External Modified Priority Search Tree.
In Section \ref{theorems} we present the two theorems that ensure the expected running times of our constructions.
The first solution is presented in Section \ref{first}, whereas our second and third constructions are discussed in Sections \ref{second} and \ref{third} respectively.
Finally, Section \ref{concl} concludes the work and briefly discusses future research in the area.

\section{Related Work and Contribution}
\label{related}

\begin{table}
 \label{3sidedtable}
 \begin{minipage}{\textwidth}
 \begin{tabular}{|r|l|l|l|l|}
   \hline
 & Model & Query Time & Update Time & Space\\
   \hline
 McCreight~\cite{MC81} & RAM &  $O(\log n + t)$ & $O(\log n)$ &$O(n)$\\
   \hline
 Willard~\cite{W00} & RAM & $O \left( \frac{\log n}{\log \log n} + t \right)$ & $O \left( \frac{\log n}{\log \log n} \right)$, $O(\sqrt{\log n})$\footnote{randomized algorithm and expected time bound}& $O(n)$\\
   \hline
   New1a\footnote{$x$ and $y$-coordinates are drawn from an unknown
$\mu$-random distribution, the $\mu$ function never changes, deletions are uniformly random over the inserted points} & RAM & $O(\log n + t)$ & $O(\log \log n)$\footnote{expected with high probability}& $O(n)$\\
   \hline
 New2a\footnote{$x$-coordinates are smoothly distributed, $y$-coordinates are arbitrarily distributed, deletions are uniformly random over the inserted points} & RAM & $O(\log \log n + t)$ $^c$& $O(\log \log n)$\footnote{amortized expected}& $O(n)$\\
   \hline
   New3a\footnote{we restrict the $x$-coordinate distribution to be $(f(n),g(n))$-\emph{smooth}, for appropriate
functions $f$ and $g$ depending on the model, and the $y$-coordinate distribution to belong to a more restricted class of
distributions. The smooth distribution is a superset of uniform and regular distributions. The restricted class contains
realistic distributions such as the Zipfian and the Power Law} & RAM & $O(\log \log n + t)$$^c$ & $O(\log \log n)$$^c$& $O(n)$\\
   \hline
   \hline
 Arge et al.~\cite{ASV99} & I/O & $O(\log_B n + t/B)$ &  $O(\log_B n)$\footnote{amortized} & $O(n/B)$\\
   \hline
 New1b$^b$ & I/O & $O(\log_B n + t/B)$ & $O(\log_B \log n)$\footnote{amortized expected w.h.p.}& $O(n/B)$\\
   \hline
 New2b$^d$ & I/O & $O(\log \log_B n + t/B)$ $^c$& $O(\log_B \log n)$ $^e$& $O(n/B)$\\
   \hline
 New3b$^f$ & I/O & $O(\log_B \log n + t/B)$ $^c$& $O(\log_B \log n)$ $^h$& $O(n/B)$\\
   \hline
 \end{tabular}
 \caption{Bounds for dynamic 3-sided planar range reporting. The number of points in the structure is $n$, the size of the query output is $t$ and the size of the block is $B$.}
 \end{minipage}
 \end{table}

The usefulness of 3-sided queries has been underlined many times in the literature \cite{CXPS+04,KRVV96}.
Apart from the significance of this query in multi-dimensional data intensive applications \cite{dKOS98,KRVV96},
3-sided queries appear in probabilistic threshold queries in uncertain databases. Such queries are studied
in a recent work of Cheng et. al. \cite{CXPS+04}.
The problem has been studied both in main memory (RAM model) and secondary storage (I/O model).
In the internal memory, the most commonly used data structure for supporting 3-sided queries is the {\em priority search tree} of McCreight~\cite{MC81}. It supports queries in $O(\log n + t)$ worst case time, insertions and deletions of points in $O(\log n)$ worst case time and uses linear space, where $n$ is the number of points and $t$ the size of the output of a query. It is a hybrid of a binary heap for the $y$-coordinates and of a balanced search tree for the $x$-coordinates.

In the static case, when points have $x$-coordinates in the set of integers $\{ 0, \ldots, n-1 \} $, the problem can be solved in $O(n)$ space and preprocessing time with $O(t+1)$ query time~\cite{ABR00}, using a {\em range minimum query} data structure~\cite{HT84} (see also Sec.~\ref{prel}).

In the RAM model, the only dynamic sublogarithmic bounds for this problem are due to Willard~\cite{W00} who attains $O \left( \log n/\log \log n \right)$ worst case or $O(\sqrt{\log n})$ randomized update time and $O \left( \log n / \log \log n + t \right)$ query time using linear space. This solution poses no assumptions on the input distribution.

Many external data structures such as grid files, various quad-trees,
z-orders and other space filling curves, k-d-B-trees, hB-trees and
various R-trees have been proposed. A recent survey can be found in
\cite{GG98}. Often these data structures are used in applications,
because they are relatively simple, require linear space and perform
well in practice most of the time. However, they all have highly sub-optimal
worst case (w.c.) performance, whereas their expected performance is usually not
guaranteed by theoretical bounds, since they are based on heuristic rules for the construction
and update operations.

Moreover, several attempts have been performed to externalize Priority Search Trees,
including \cite{BG90}, \cite{IKO87}, \cite{KRVV96}, \cite{RS94} and \cite{SR95},
but all of them have not been optimal. The worst case optimal external memory
solution (External Priority Search Tree) was presented in \cite{ASV99}.
It consumes $O(n/B)$ disk blocks, performs 3-sided range queries in $O(\log_{B}n + t/B)$ I/Os
w.c. and supports updates in $O(\log_{B}n)$ I/Os amortized. This solution poses no assumptions on the input distribution.

In this work, we present new data structures for the RAM and the I/O model that improve by a logarithmic factor
the update time in an expected sense and attempt to improve the query complexity likewise.
The bounds hold with high probability (w.h.p.) under assumptions on the distributions of the input coordinates.
We propose three multi-level solutions, each with a main memory and an external memory variant.

For the first solution, we assume that the $x$ and $y$ coordinates are being continuously drawn from an unknown $\mu$-random
distribution. It consists of two levels, for both internal and external variants. The upper level of the first solution consists of a single Priority Search Tree \cite{MC81} that indexes the structures of the lower level. These structures are Priority Search Trees as well. For the external variant we substitute the structures with their corresponding optimal external memory solutions, the External Priority Search Trees \cite{ASV99}. The internal variant achieves $O(\log n + t)$ w.c. query time and $O(\log \log n)$ expected w.h.p. update time, using linear space. The external solution attains $O(\log_B n +t/B)$ I/Os w.c. and $O(\log_B \log n)$ I/Os amortized expected w.h.p. respectively, and uses linear space.

For the second solution, we consider the case where the $x$-coordinates of inserted points are drawn from a {\em smooth} probabilistic distribution, and the $y$-coordinates are arbitrarily distributed. Moreover, the deleted points are selected uniformly at random among the points in the data structure and queries can be adversarial. The assumption on the $x$-coordinates is broad enough to include distributions used in practice, such as uniform, regular and classes of non-uniform ones~\cite{AM93,KMSTTZ03}. We present two linear space data structures, for the RAM and the I/O model respectively. In the former model, we achieve a query time of $O(\log \log n + t)$ expected with high probability and update time of $O(\log \log n)$ expected amortized. In the latter model, the I/O complexity is $O(\log \log_B n + t/B)$ expected with high probability for the query and $O(\log_B \log n )$ expected amortized for the updates. In both cases, our data structures are deterministic and the expectation is derived from a probabilistic distribution of the $x$-coordinates, and an expected analysis of updates of points with respect to their $y$-coordinates.

By the third solution, we attempt to improve the expected query complexity and simultaneously preserve the update and space
complexity. In order to do that, we restrict the $x$-coordinate distribution to be $(f(n),g(n))$-\emph{smooth}, for appropriate
functions $f$ and $g$ depending on the model, and the $y$-coordinate distribution to belong to a more restricted class of
distributions. The smooth distribution is a superset of uniform and regular distributions. The restricted class contains
realistic distributions such as the Zipfian and the Power Law. The internal variant consists of two levels, of which the lower level is identical to that of the first solution. We implement the upper level with a static Modified Priority Search Tree \cite{SMKLTTV04}. For the external variant, in order to achieve the desired bounds, we introduce three levels. The lower level is again identical to that of the first solution, while the middle level consists of $O(B)$ size buckets. For the upper level we use an External Modified Priority Search Tree, introduced here for the first time. The latter is a straight forward externalization of the Modified Priority Search Tree and is static as well. In order to make these trees dynamic we use the technique of global rebuilding \cite{LO88}. The internal version reduces the query complexity to $O(\log \log n + t)$ expected with high probability and the external to $O(\log_B \log n + t/B)$ I/Os expected with high probability. The results are summarized in Table 1. 

\section{Data Structures and Probability Distributions}
\label{prel}
For the main memory solutions we consider the RAM model of computation.
We denote by $n$ the number of elements that reside in the data structures and by $t$ the size of the query.
The universe of elements is denoted by $S$. When we mention that a data structure performs an
operation in an \emph{amortized expected with high probability} complexity, we mean the bound is expected to
be true with high probability, under a worst case sequence of insertions and deletions of points.

For the external memory solutions we consider the I/O model of computation \cite{V01}. That means that the input resides in the
external memory in a blocked fashion. Whenever a computation needs to be performed to an element, the block of size $B$ that
contains that element is transferred into main memory, which can hold at most $M$ elements. Every computation that is performed in
main memory is free, since the block transfer is orders of magnitude more time consuming. Unneeded blocks that reside in the main
memory are evicted by a LRU replacement algorithm. Naturally, the number of block transfers (\emph{I/O operation}) consists the
metric of the I/O model.

Furthermore, we will consider that the points to be inserted are continuously drawn by specific distributions, presented in
the sequel. The term \emph{continuously} implies that the distribution from which we draw the points remains unchanged. Since the
solutions are dynamic, the asymptotic bounds are given with respect to the current size of the data structure. Finally, deletions
of the elements of the data structures are assumed to be uniformly random. That is, every element present in the data structure is
equally likely to be deleted \cite{K77}.

\subsection{Probability Distributions}
In this section, we overview the probabilistic distributions that will be used in the remainder of the paper.
We will consider that the $x$ and $y$-coordinates are distinct elements of these distributions and will choose
the appropriate distribution according to the assumptions of our constructions.

A probability distribution is \emph{$\mu$-random} if the elements are drawn randomly with respect
to a density function denoted by $\mu$. For this paper, we assume that $\mu$ is unknown.

Informally, a distribution defined over an interval $I$ is
\textit{smooth} if the probability density over any subinterval of $I$
does not exceed a specific bound, however small this subinterval is
(i.e., the distribution does not contain sharp peaks). Given two functions $f_{1}$ and $f_{2}$, a density function
$\mu=\mu[a,b](x)$ is {\em $(f_{1},f_{2})$-smooth} \cite{MT93, AM93} if there
exists a constant $\beta$, such that for all $c_{1},c_{2},c_{3}$, $a
\leq c_{1} < c_{2} < c_{3} \leq b$, and all integers $n$, it holds that:

$$ \int^{c_{2}}_{c_{2}-\frac{c_{3}-c_{1}}{f_{1}(n)}}
  {\mu[c_{1},c_{3}](x)dx} \leq \frac{\beta \cdot f_{2}(n)}{n} $$

\noindent
where $\mu[c_{1},c_{3}](x)=0$ for $x<c_1$ or $x>c_3$, and
$\mu[c_{1},c_{3}](x)=\mu(x)/p$ for $c_1\le x\le c_3$ where
$p=\int_{c_1}^{c_3}\mu(x) dx$. Intuitively, function $f_1$ partitions
an arbitrary subinterval $[c_1,c_3] \subseteq [a,b]$ into $f_1$ equal
parts, each of length $\frac{c_{3}-c_{1}}{f_1}= O(\frac{1}{f_1})$;
that is, $f_1$ measures how fine is the partitioning of an arbitrary
subinterval. Function $f_2$ guarantees that no part, of the $f_1$
possible, gets more probability mass than $\frac{\beta \cdot f_2}{n}$;
that is, $f_2$ measures the sparseness of any subinterval
$[c_{2}-\frac{c_{3}-c_{1}}{f_1},c_2] \subseteq [c_1,c_3]$. The class
of $(f_{1},f_{2})$-smooth distributions (for appropriate choices of
$f_1$ and $f_2$) is a superset of both regular and uniform classes of
distributions, as well as of several non-uniform classes
\cite{AM93,KMSTTZ03}. Actually, {\em any} probability distribution is
$(f_{1},\Theta(n))$-smooth, for a suitable choice of $\beta$.

The \emph{grid distribution} assumes that the elements are integers that belong to a specific range $[1,M]$.

We define the \emph{restricted class} of distributions as the class that contains distributions used in practice,
such as the Zipfian, Power Law, e.t.c..

The \emph{Zipfian} distribution is a distribution of probabilities of
occurrence that follows Zipf's law. Let $N$ be the number of elements,
$k$ be their rank and $s$ be the value of the exponent characterizing
the distribution. Then Zipf's law is defined as the function $f(k;s,N)
= \frac{1/k^s}{\Sigma_{k=1}^N 1/n^s}$. Intuitively, few elements occur
very often, while many elements occur rarely.

The \emph{Power Law} distribution is a distribution over probabilities that
satisfy $Pr[X \geq x] = cx^{-b}$ for constants $c, b > 0$.

\subsection{Data Structures}
\label{DS}
In this section, we describe the data structures that we will combine
in order to achieve the desired complexities.

\subsubsection{Priority Search Trees:}
The classic \emph{Priority Search Tree (PST)} \cite{MC81} stores
points in the 2-d space. One of the most important operations that the
PST supports is the {\em 3-sided query}. The 3-sided query consists of
a half bounded rectangle $[a,b]\times(-\infty,c]$ and asks for
all points that lie inside this area. Note that by rotation we can
unbound any edge of the rectangle. The PST supports this operation in
$O(\log{n}+t)$ w.c., where $n$ is the number of points and $t$ is the
number of the reported points.

The PST is a combination of a search tree and a priority queue. The
search tree (an $(a,b)$-tree suffices) allows the efficient support of
searches, insertions and deletions with respect to the $x$-coordinate,
while the priority queue allows for easy traversal of points with
respect to their $y$-coordinate. In particular, the leaves of the PST
are the points sorted by $x$-coordinate. In the internal nodes of the
tree there are artificial values which are used for the efficient
searching of points with respect to their $x$-coordinate. In addition,
each internal node stores a point that has the minimum $y$-coordinate
among all points stored in its subtree. This corresponds to a
tournament on the leaves of the PST. For example, the root of the PST
contains a point which has minimum $y$-coordinate among all points in
the plane, as well as a value which is in the interval defined between
the $x$-coordinates of the points stored in the rightmost leaf of the
left subtree and the leftmost leaf of the right subtree (this is true
in the case of a binary tree). A PST implemented with an red-black tree
supports the operations of insertion of a new point, deletion of an
existing point and searching for the $x$-coordinate of a point in
$O(\log{n})$ worst case time.

Regarding the I/O model, after several attempts, a worst case optimal
solution was presented by Arge et al. in \cite{ASV99}. The proposed
indexing scheme consumes $O(n/B)$ space, supports updates in $O(\log_B
n)$ amortized I/Os and answers 3-sided range queries in $O(\log_B n + t/B)$ I/Os. We
will refer to this indexing scheme as the \emph{External Priority Search Tree
(EPST)}.

\subsubsection{Interpolation Search Trees:}

In \cite{KMSTTZ06}, a dynamic data structure based on interpolation search (IS-Tree) was presented,
which consumes linear space and can be updated in $O(1)$ time w.c. Furthermore, the elements can be searched
in $O(\log \log n)$ time expected w.h.p., given that they are drawn from a $(n^{\alpha}, n^{\beta})$-smooth distribution,
for any arbitrary constants $0< \alpha , \beta < 1$. The externalization of this data structure,
called interpolation search B-tree (ISB-tree), was introduced in \cite{KMMSTTZ05}. It supports update operations in $O(1)$
worst-case I/Os provided that the update position is given and search operations in $O(\log_B \log n)$ I/Os expected w.h.p.
The expected search bound holds w.h.p. if the elements are drawn by a $(n/(\log \log n)^{1+\epsilon}, n^{1-\delta})$-smooth
distribution, where $\epsilon > 0$ and $\delta = 1 -\frac{1}{B}$ are constants. The worst case search bound
is $O(\log_{B}n)$ block transfers.

\subsubsection{Weight Balanced Exponential Tree:}
\label{wbest}
The {\em exponential search tree} is a technique for converting static polynomial space search structures for ordered sets into fully-dynamic linear space data structures. It was introduced in~\cite{A96,T98,AT07} for searching and updating a dynamic set $U$ of $n$ integer keys in linear space and optimal $O$$($$\sqrt{\smash\log n \smash /\smash\log\smash\log n}$$)$ time in the RAM model. Effectively, to solve the dictionary problem, a doubly logarithmic height search tree is employed that stores static local search structures of size polynomial to the degree of the nodes. 

Here we describe a variant of the exponential search tree that we dynamize using a rebalancing scheme relative to that of the {\em weight balanced search trees}~\cite{AV96}. In particular, a {\em weight balanced exponential tree} $T$ on $n$ points is a leaf-oriented rooted search tree where the degrees of the nodes increase double exponentially on a leaf-to-root path. All leaves have the same depth and reside on the lowest level of the tree (level zero). The {\em weight} of a subtree $T_u$ rooted at node $u$ is defined to be the number of its leaves. If $u$ lies at level $i\geq 1$, the weight of~$T_u$ ranges within $\left[ \frac{1}{2} \cdot w_{i} + 1 , 2 \cdot w_{i} -1 \right]$, for a {\em weight parameter} $w_i = c_1^{\smash {c_2^{i}}}$ and constants $c_2>1$ and $c_1\geq 2^{3 / (c_2 -1)}$ (see Lem.~\ref{rebal}). Note that $w_{i+1} = w_i^{c_2}$. The root does not need to satisfy the lower bound of this range. 
The tree has height~$\Theta (\log_{c_2}\log_{c_1} n)$.

The insertion of a new leaf to the tree increases the weight of the nodes on the leaf-to-root path by one. This might cause some weights to exceed their range constraints (``overflow''). We {\em rebalance} the tree in order to revalidate the constraints by a leaf-to-root traversal, where we ``split'' each node that overflowed. An overflown node $u$ at level $i$ has weight $2 w_i$. A split is performed by creating a new node $v$ that is a sibling of $u$ and redistributing the children of $u$ among $u$ and $v$ such that each node acquires a weight within the allowed range. In particular, we scan the children of $u$, accumulating their weights until we exceed the value~$w_i$, say at child $x$. Node $u$ gets the scanned children and $v$ gets the rest. Node~$x$ is assigned as a child to the node with the smallest weight. Processing the overflown nodes $u$ bottom up guarantees that, during the split of~$u$, its children satisfy their weight constraints.

The deletion of a leaf might cause the nodes on the leaf-to-root path to ``underflow'', i.e. a node $u$ at level $i$ reaches weight $\frac{1}{2} w_i$. By an upwards traversal of the path, we discover the underflown nodes. In order to revalidate their node constraints, each underflown node chooses a sibling node $v$ to ``merge'' with. That is, we assign the children of $u$ to $v$ and delete $u$. Possibly, $v$ needs to ``split'' again if its weight after the merge is more than $\frac{3}{2} w_i$ (``share''). In either case, the traversal continues upwards, which guarantees that the children of the underflown nodes satisfy their weight constraints. The following lemma, which is similar to~\cite[Lem. 9]{AV96}, holds.
%

\begin{lemma}
\label{rebal}
After rebalancing a node $u$ at level $i$, $\Omega(w_i)$ insertions or deletions need to be performed on $T_u$, for $u$ to overflow or underflow again.
\end{lemma}
\begin{proof}
A split, a merge or a share on a node $u$ on level $i$ yield nodes with weight in $\left[ \frac{3}{4}w_i - w_{i-1}, \frac{3}{2}w_i + w_{i-1} \right]$. If we set $w_{i-1}\leq \frac{1}{8}w_i$, which always holds for $c_1\geq 2^{3/(c_2 -1)}$, this interval is always contained in $[ \frac{5}{8}w_i, \frac{14}{8}w_i]$.
\qed
\end{proof}


\begin{figure}
\label{fig:mpst}
\includegraphics[scale=0.70]{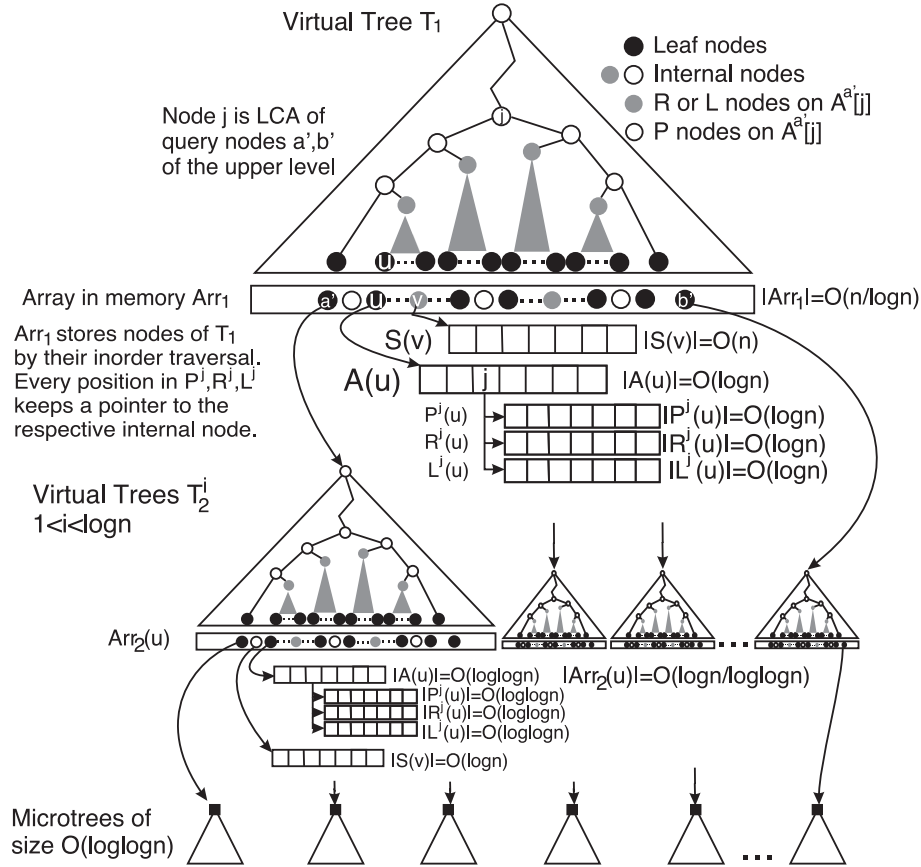}
\caption{The linear space MPST.}
\end{figure}

\subsubsection{Range Minimum Queries}
\label{RMQ}
The {\em range minimum query} (RMQ) problem asks to preprocess an array of size $n$ such that, given an index range, one can report the position of the minimum element in the range. In~\cite{HT84} the RMQ problem is solved in $O(1)$ time using $O(n)$ space and preprocessing time. The currently most space efficient solution that support queries in $O(1)$ time appears in~\cite{FH07}.

\subsubsection{Dynamic External Memory 3-sided Range Queries for $O(B^2)$ Points.}
For external memory, Arge et al.~\cite{ASV99} present the following lemma for handling a set of at most $B^2$ points.

\begin{lemma}
\label{rebal}
A set of $K \leq B^2$ points can be stored in $O(K/B)$ blocks, so that 3-sided queries need $O(t/B + 1)$ I/Os and updates $O(1)$ I/Os, for output size~$t$
\end{lemma}

\begin{proof}
See Lemma 1 presented in \cite{ASV99}.
\end{proof}

\subsubsection{Modified Priority Search Trees}
A \emph{Modified Priority Search Tree (MPST)} is a static data structure that stores points on the plane
and supports 3-sided queries. It is stored as an array ($Arr$) in memory, yet it can be visualized as a complete binary tree.
Although it has been presented in \cite{KMSTTV00}, \cite{SMKLTTV04} we sketch it here again, in order to introduce its external version.

Let $T$ be a Modified Priority Search Tree (MPST) \cite{SMKLTTV04}
which stores $n$ points of $S$ (see figure 2). We denote by $T_v$ the subtree of
$T$ with root $v$. Let $u$ be a leaf of the tree. Let $P_u$ be the root-to-leaf path for
$u$. 
For every $u$, we sort the points in $P_u$ by their
$y$-coordinate. We denote by $P^j_u$ the subpath of $P_u$ with nodes
of depth bigger or equal to $j$ (The depth of the root is 0). Similarly $L^j_u$ (respectively
$R^j_u$ ) denotes the set of nodes that are left (resp. right)
children of nodes of $P^j_u$ and do not belong to $P^j_u$ . The tree structure $T$ has the following
properties:

\begin{itemize}
\item Each point of $S$ is stored in a leaf of
$T$ and the points are in sorted x-order from left to right.
\item Each internal node $v$ is equipped with a
secondary list $S(v)$. $S(v)$ contains in the points
stored in the leaves of $T_v$ in increasing y-coordinate.
\item A leaf $u$ also stores the following lists $A(u)$, $P^{j}(u)$, $L^j(u)$ and $R^j(u)$, for $0\leq j \leq \log n$.
The list $P^j(u)$, $L^j(u)$ and
$R^j(u)$ store, in increasing $y$-coordinate, pointers to the
respective internal nodes. $A(u)$ is an array that indexes $j$.
\end{itemize}

Note that the first element of the list $S(v)$ is the point of the
subtree $T_v$ with minimum $y$-coordinate. Also note that $0\leq j \leq \log n$,
so there are $\log n$ such sets $P^{j}_u$, $L^j_u$, $R^j_u$ for each leaf $u$.
Thus the size of $A$ is $\log n$ and for a given $j$, any list $P^{j}(u)$, $L^j(u)$ or $R^j(u)$ can be accessed in constant time.
By storing the nodes of the tree $T$ according to their inorder traversal in an array $Arr$ of size $O(n)$,
we can imply the structure of tree $T$. Also each element of $Arr$ contains a binary label that corresponds
to the inorder position of the respective node of $T$, in order to facilitate constant time lowest common ancestor (LCA) queries.

To answer a query with the range $[a, b]\times (-\infty, c]$ we find
the two leaves $u$, $w$ of $Arr$ that contain $a$ and $b$
respectively. If we assume that the leaves that contain $a,b$ are
given, we can access them in constant time.
Then, since $Arr$ contains an appropriate binary label, we use a simple LCA
(Lowest Common Ancestor) algorithm \cite{G94,HT84} to compute the
depth $j$ of the nearest common ancestor of $u$, $w$ in $O(1)$
time. That is done by performing the XOR operation between the binary labels of
the leaves $u$ and $w$ and finding the position of the
first set bit provided that the left-most bit is placed in position
$0$. Afterwards, we traverse $P^{j}(u)$ until the scanned
$y$-coordinate is not bigger than $c$. Next, we traverse $R^{j}(u)$,
$L^{j}(w)$ in order to find the nodes whose stored points have
$y$-coordinate not bigger than $c$. For each such node $v$ we traverse the
list $S(v)$ in order to report the points of $Arr$ that satisfy the
query. Since we only access points that lie in the query, the total query time is $O(t)$, where $t$ is the answer size.

The total size of the lists $S(u)$ for each level of $T$ is $O(n)$. Each of the $O(n)$ leaves stores $\log n$ lists $P_j$,
$L_j$ and $R_j$, each of which consumes $O(\log n)$ space. Thus the
space for these lists becomes $O(n \log^2 n)$. By implementing these
lists as partially persistent sorted lists \cite{BGOSW96}, their total space becomes $O(n \log n)$, resulting in
a total space of $O(n \log n)$ for these lists. Thus, the total space
occupied by $T$ is $O(n \log n)$.

We can reduce the space of the structure by pruning as in
\cite{FMNT87,O88}. However, pruning alone does not reduce the space to
linear. We can get better but not optimal results by applying pruning
recursively. To get an optimal space bound we will use a combination
of pruning and table lookup. The pruning method is as follows: Consider the nodes of $T$, which
have height $\log \log n$. These nodes are roots of subtrees of $T$ of
size $O(\log n)$ and there are $O(n/ \log n)$ such nodes. Let $T_1$ be
the tree whose leaves are these nodes and let $T^i_2$ be the subtrees
of these nodes for $1 \leq i \leq O(n/ \log n)$. We call $T_1$ the first layer of the structure and
the subtrees $T^i_2$ the second layer. $T_1$ and each subtree $T^i_2$
is by itself a Modified Priority Search Tree. Note that $T_1$ has
size $O(n/ \log n) = O(n)$. Each subtree $T^i_2$ has $O(\log n/ \log \log n)$
leaves and depth $O(\log \log n)$. The space for the second layer is $O(n \log n)$.
By applying the pruning method to all the trees of the second layer we
get a third layer which consists of $O(n / \log \log n)$ modified
priority search trees each of size $O(\log \log n)$. Ignoring the third layer, the second layer needs now linear space,
while the $O(n \log n)$ space bottleneck is charged on the third level. If we use table lookup \cite{GT85} to implement
the modified priority search trees of the third layer we can reduce its space to linear, thus consuming linear space in total.

In order to answer a query on the three layered structure we access
the microtrees that contain $a$ and $b$ and extract in $O(1)$ time the
part of the answer that is contained in them. Then we locate the
subtrees $T^i_2$, $T^j_2$ that contain the representative leaves of
the accessed microtrees and extract the part of the answer that is
contained in them by executing the query algorithm of the MPST. The
roots of these subtrees are leaves of $T_1$. Thus we execute again the
MPST query algorithm on $T_1$ with these leaves as arguments. Once we reach the node with $y$-coordinate bigger than $c$,
we continue in the same manner top down. This may lead us to subtrees of the
second layer that contain part of the answer and have not been
accessed yet. That means that for each accessed tree of the second layer, we execute the MPST query algorithm,
where instead of $a$ and $b$, we set as arguments the minimum and the maximum $x$-coordinates of all the points stored
in the queried tree. The argument $c$ remains, of course, unchanged. Correspondingly, in that way we access the microtrees
of the third layer that contain part of the answer. We execute the top down part of the algorithm on them,
in order to report the final part of the answer.

\begin{lemma}
Given a set of $n$ points on the plane we can store them in a static data
structure with $O(n)$ space that allows three-sided range queries to
be answered in $O(t)$ worst case, where $t$ is the answer size.
\end{lemma}

\begin{proof}
See \cite{SMKLTTV04}.
\end{proof}

The \emph{External Modified Priority Search Tree (EMPST)} is similar to the MPST, yet we store the lists in a blocked fashion.
In order to attain linear space in external memory we prune the structure $k$ times, instead of two times.
The pruning terminates when $\log^{(k)} n = O(B)$. Since computation within a block is free, we do not need the additional
layer of microtrees. By that way we achieve $O(n/B)$ space.

Assume that the query algorithm accesses first the two leaves $u$ and $v$ of the $k$-th layer of the EMPST,
which contain $a$ and $b$ respectively. If they belong to different EMPSTs of that layer, we recursively take the roots
of these EMPSTs until the roots $r_u$ and $r_v$ belong to the same EMPST, w.l.o.g. the one on the upper layer.
That is done in $O(k)=O(1)$ I/Os. Then, in O(1) I/Os we access the $j$-th entry of $A(r_u)$ and $A(r_v)$,
where $j$ is the depth of $LCA(r_u,r_v)$, thus also the corresponding sublists $P^j(r_u), R^j(r_u), L^j(r_u)$
and $P^j(r_v), R^j(r_v), L^j(r_v)$. Since the\-se sublists are $y$-ordered, by scanning them in $t_1/B$ I/Os
we get all the $t_1$ pointers to the $S$-lists that contain part of the answer. We access the $S$-lists in $t_1$ I/Os
and scan them as well in order to extract the part of the answer (let's say $t_2$) they contain.
We then recursively access the $t_2$ $S$-lists of the layer below and extract the part $t_3$ that resides on them.
In total, we consume $t_1/B + t_1\cdot t_2/B + ... + t_{i-1} \cdot t_{i}/B + ... + t_{k-1} \cdot t_{k}/B$ I/Os.
Let $p_i$ the probability that $t_i=t^{p_i}$ where $t$ is the total size of the answer and
$\sum_{i=1}^k p_i=1$. Thus, we need $t^{p_1}/B + \sum_{i=1}^{k-1} \frac{t^{p_i}}{B}\cdot t^{p_{i+1}}$ I/Os or
$t^{p_1}/B + \sum_{i=1}^{k-1} \frac{t^{(p_i+p_{i+1})}}{B}$ I/Os.
Assuming w.h.p. an equally likely distribution of answer amongst the $k$ layers,
we need $t^{\frac{1}{k}}/B + \sum_{i=1}^{k-1} \frac{t^{\frac{1}{k}+ \frac{1}{k}}}{B}$ expected number
of I/Os or $t^{\frac{1}{k}}/B + \sum_{i=1}^{k-1} \frac{t^{\frac{2}{k}}}{B}$.
Since $k>>2$, we need totally $O(t/B)$ expected w.h.p. number of I/Os.

\begin{lemma}
Given a set of $n$ points on the plane we can store them in a static data
structure with $O(n/B)$ space that allows three-sided range queries to
be answered in $O(t/B)$ expected w.h.p. case, where $t$ is the size of the
answer.
\end{lemma}

\section{Expected First Order Statistic of Uknown Distributions}
\label{theorems}
In this section, we prove two theorems that will ensure the expected running times of our constructions.
They are multilevel data structures, where for each pair of levels, the upper level indexes representative elements
(in our case, point on the plane) of the lower level buckets.
We call an element \emph{violating} when its insertion to or deletion from the lower level bucket causes
the representative of that bucket to change, thus triggering an update on the upper level.
We prove that for an epoch of $O(\log n)$ updates, the number of violating elements is $O(1)$ if they are
continuously being drawn from a $\mu$-random distribution. Secondly, we prove that for a broader epoch of $O(n)$ updates,
the number of violating elements is $O(\log n)$, given that the elements are being continuously drawn from a distribution
that belongs to the restricted class. Violations are with respect to the $y$-coordinates, while the distribution
of elements in the buckets are with respect to $x$-coordinates.

But first, the proof of an auxiliary lemma is necessary. Assume a sequence $\mathcal{S}$ of distinct numbers generated
by a continuous distribution $\mu = \mathcal{F}$ over a universe $\mathcal{U}$. Let $|\mathcal{S}|$ denote the size
of $\mathcal{S}$. Then, the following holds:

\begin{lemma} \label{lmm:min}
The probability that the next element $q$ drawn from $\mathcal{F}$ is less than the minimum element $s$ in
$\mathcal{S}$ is equal to $\frac{1}{|\mathcal{S}|+1}$.
\end{lemma}

\begin{proof}
Suppose that we have $n$ random observations $X_1$,\ldots,$X_n$ from
an unknown continuous probability density function $f(X)$, with cumulative
distribution $\mu = F(X)$, $X \in \left[ a, b \right]$. We want to
compute the probability that the $(n+1)-th$ observation is less than
$min\left\{X_1,\ldots,X_n \right\}$.
Let $X_{(1)}=min\left\{X_1,\ldots,X_n \right\}$. Therefore,
$P\left\{ X_{n+1} < X_{(1)} \right\}$ = $\sum_{x} P\left\{ X_{n+1} < X_{(1)} / X_{(1)}=x
\right\} \cdot$ $P \left\{ X_{(1)}=x \right\}$ $\left( \alpha \right)$.

It is easy to see that $P\left\{ X_{n+1} < X_{(1)} / X_{(1)} = x \right\}$ = $F(X)$ =
$P \left\{ X_{n+1} < x \right\}$ $\left( \beta \right)$.
Also $P \left\{ X_{(k)}=x \right\}=n \cdot f(x) \cdot \left(
^{n-1}_{k-1} \right) \cdot {F(X)}^{k-1} \cdot (1-F(X))^{n-k}$ $\left(
\gamma \right)$, where $X_{(k)}$ is the $k-th$ smallest value in
$\left\{ X_1,\ldots,X_n \right\}$.

In our case $k=1$, which intuitively means that we have $n$ choices for one in $\left\{
X_1,\ldots,X_n \right\}$ being the smallest value. This is true if
all the rest $n-1$ are more than $x$, which occurs with probability:
$\left( 1-F(X) \right)^{n-1}=\left( 1-P\left\{ X<x \right\}
\right)^{n-1}$.
By $\left( \beta \right)$ and $\left( \gamma \right)$, expression
$\left( \alpha \right)$ becomes:



\noindent
$P \left\{ X_{n+1} < X_{(1)} \right\}$ = $\int^{a}_{b} n \cdot f(X)\left(
^{n-1}_{k-1} \right) \cdot F(X) \cdot \left( 1-F(X) \right)^{n-1}$ $dX$.
After some mathematical manipulations, we have that:

\noindent
$P \left\{ X_{n+1} < X_{(1)} \right\}=\int^{a}_{b} n \cdot f(X) \cdot
\left( 1-F(X) \right)^{n-1} \cdot F(X)dX= \\
 \int^{a}_{b} {\left[ -
\left(1-F(X)\right)^n \right]}^{'}F(X)$ $dX$ = $\int^{a}_{b} {\left[ -
\left(1-F(X)\right)^n \cdot F(X)\right]}^{'}$ $dX$ $+$ $\int^{a}_{b}
\left(1-F(X)\right)^n \cdot F^{'}(X)dX= \left\{ -\left( 1-F(X)
\right)^n \cdot F(x)|^{b}_{a} \right\}$ $+$ $\int^{a}_{b} -\left[
\frac{\left( 1-F(X) \right)^{n+1}}{n+1} \right]^{'}dX= -\left( 1-F(b)
\right)^n \cdot F(b)+\left( 1-F(a) \right)^n \cdot F(a) - \left\{
\frac{\left( 1-F(X) \right)^{n+1}}{n+1} |^{b}_{a} \right\}= -\left\{
\frac{\left( 1-F(b) \right)^{n+1}}{n+1}-\frac{\left( 1-F(a)
\right)^{n+1}}{n+1}\right\} =\frac{1}{n+1}$
\end{proof}

Apparently, the same holds if we want to maintain the maximum element of the set $\mathcal{S}$.

\begin{prop} \label{Prop:randomness-invarince}
Suppose that the input elements have their $x$-coordinate generated by an arbitrary continuous distribution
$\mu$ on $[a, b] \subseteq \Re$.
Let $n$ be the elements stored in the data structure at the latest reconstruction.
An epoch starts with $\log n$ updates. During the $i$-th update let
$N(i) \in [n, r\cdot n]$, with constant $r > 1$, denote the number of elements currently
stored into the $\frac{n}{\log n}$ buckets that partition $[a, b] \subseteq \Re$.
Then the $N(i)$ elements remain $\mu$ randomly distributed in the buckets per $i$-th update.
\end{prop}

\begin{proof}
The proof is analogous to \cite[Lem. 2]{KMSTTZ03} and is omitted.
\end{proof}

\begin{theorem} \label{thm:insdel:logn}
For a sequence of  $O( \log n)$ updates, the expected number of violating elements is $O(1)$,
assuming that the elements are being continuously drawn from a $\mu$-random distribution.
\end{theorem}
\begin{proof}
According to Prop. \ref{Prop:randomness-invarince}, there are $N(i) \in [n, r\cdot n]$ (with constant $r > 1$)
elements with their  $x$-coordinates $\mu$-randomly distributed in the buckets $j=1,\ldots,\frac{n}{\log n}$,
that partition $[a, b] \subseteq \Re$. By \cite[Th. 4]{KMSTTZ03}, with high probability,
each bucket $j$ receives an $x$-coordinate with probability $p_j= \Theta(\frac{\log n}{n})$.
It follows that during the $i$-th update operation, the elements in bucket $j$ is a Binomial random variable with
mean $p_j \cdot N(i)= \Theta(\log n)$.

The elements with $x$-coordinates in an arbitrary bucket $j$ are $\alpha N(i)$ with
probability
${N(i) \choose  \alpha
N(i)} p_j^{\alpha N(i)} (1-p_j)^{(1- \alpha)N(i)}
\sim
\left[
\left(\frac{p_j}{\alpha}\right)^{\alpha}
\left(\frac{1-p_j}{1-\alpha}\right)^{1-\alpha}
\right]^{N(i)}.
$
In turn, these are $\leq \alpha N(i)=\frac{p_j }{2} N(i)$ (less than half of the bucket's mean) with probability
\begin{eqnarray}
\leq \frac{p_j N(i)}{2} \cdot
\left[
\left(\frac{p_j}{\alpha}\right)^{\alpha}
\left(\frac{1-p_j}{1-\alpha}\right)^{1-\alpha}
\right]^{N(i)}\rightarrow 0
\label{Eq:pr-half-expectation}
\end{eqnarray}
as $n\rightarrow \infty$ and $\alpha = \frac{p_j }{2}$.

Suppose that an element is inserted in the $i$-th update. It induces a violation if its $y$-coordinate
is strictly the minimum element of the bucket $j$ it  falls into.
\begin{itemize}
\item
If the bucket contains $\geq \frac{p_j }{2}\log N(i) \geq \frac{p_j }{2}\log n$ coordinates then by Lemma \ref{lmm:min}
element $y$ incurs a violation with probability $O(\frac{1}{\log n})$.
\item
If the bucket contains $< \frac{p_j }{2}\log N(i)$ coordinates, which is as likely as in Eq. (\ref{Eq:pr-half-expectation}),
then element $y$ may induce $\leq 1$ violation.
\end{itemize}
Putting these cases together, element $y$ expectedly induces at most $O(\frac{1}{\log n}) +$Eq. (\ref{Eq:pr-half-expectation})$ = O(\frac{1}{\log n})$ violations.
We conclude that during the whole epoch of $\log n$ insertions the expected number of violations are at most
 $\log n \cdot O(\frac{1}{\log n})$ plus $\log n \cdot$ Eq. (\ref{Eq:pr-half-expectation}) which is $O(1)$.
\end{proof}

\begin{theorem} \label{thm:insdel:n}
For a sequence of  $O(n)$ updates, the expected number of violating elements is $O(\log n)$,
assuming that $x-$ coordinates are drawn from a continuous smooth distribution and the $y-$ coordinates
are drawn from the restricted class of distributions (\emph{power-law} or \emph{zipfian}).
\end{theorem}
\begin{proof}
Suppose an element is inserted, with its $y$-coordinate
following a discrete distribution (while its $x$-coordinate is arbitrarily distributed) in the universe  $\{y_1, y_2, \ldots \}$ with $y_i < y_{i+1}, \forall i \geq 1$.
Also, let $q= \Pr [y > y_1]$ and $y_j^*$ the $\min$ $y$-coordinate of the elements in bucket $j$ as soon as
the current epoch starts.
Clearly, the element just inserted incurs a violation when landing into bucket $j$ with probability $\Pr[y < y_j^*]$.
\begin{itemize}
\item
If the bucket contains $\geq \frac{p_j }{2}\log N(i) \geq \frac{p_j }{2}\log n$ coordinates, then
coordinate $y$ incurs a violation with probability $\leq q^{\frac{p_j }{2}\log n}$.
(In other words, a violation may happens when at most all the $\Omega (\log n)$ coordinates of the elements in bucket
$j$ are $> y_1$, that is, when $y_j^* > y_1$.)
\item
If the bucket contains $< \frac{p_j }{2}\log N(i)$ coordinates, which is as likely as in Eq. (\ref{Eq:pr-half-expectation})
then coordinate $y$ may induces $\leq 1$ violation.
\end{itemize}
All in all, $y$ coordinate expectedly induces $\leq q^{\Omega (\log n)} +$ Eq. (\ref{Eq:pr-half-expectation}) violations.
Thus, during the whole epoch of $n$  insertions the expected number of violations are at most
$n \cdot \left(q^{\Omega (\log n)}\right) + n \cdot$  Eq. (\ref{Eq:pr-half-expectation}) $= n q^{\Omega (\log n)} + o(1)$ violations.
This is at most  $c \cdot \log n= O(\log n)$ if $q \leq \left(\frac{c\log n}{n} \right)^{(\log n)^{-1}} \rightarrow e^{-1}$ as $n \rightarrow \infty$.
\end{proof}
\begin{rem}
Note that {\em Power Law} and {\em Zipfian} distributions have the aforementioned property that $q \leq \left(\frac{c\log n}{n} \right)^{(\log n)^{-1}}
\rightarrow e^{-1}$ as $n \rightarrow \infty$.
\end{rem}

\section{The First Solution for Random Distributions}
\label{sec:str}
\label{first}
In this section, we present the construction that works under the assumptions that the $x$ and $y$-coordinates are continuously
drawn by an unknown $\mu$-random distribution.

The structure we propose consists of two levels, as well as an
auxiliary data structure. All of them are implemented as PSTs.
The lower level partitions the points into {\em buckets} of almost
equal logarithmic size according to the $x$-coordinate of the
points. That is, the points are sorted in increasing order according
to $x$-coordinate and then divided into sets of $O(\log{n})$
elements each of which constitutes a bucket. A bucket $C$ is implemented as a PST and is
represented by a point $C^{min}$ which has the smallest $y$-coordinate
among all points in it. This means that for each bucket the cost for insertion, deletion and search
is equal to $O(\log{\log{n}})$, since this is the height of the PST
representing $C$.


The upper level is a PST on the representatives of the lower
level. Thus, the number of leaves in the upper level is $O\left(
\frac{n}{\log{n}} \right)$. As a result, the upper level supports the
operations of insert, delete and search in $O(\log{n})$ time. In
addition, we keep an extra PST for insertions of violating points. Under this context, we
call a point $p$ {\em violating}, when its $y$-coordinate is less than $C^{min}$
of the bucket $C$ in which it should be inserted. In the case of a
violating point we must change the representative of $C$ and as a
result we should make an update operation on the PST of the upper
level, which costs too much, namely $O(\log{n})$.

We assume that the $x$ and $y$-coordinates are drawn from an unknown
$\mu$-random distribution and that the $\mu$ function never changes.
Under this assumption, according to the combinatorial game of bins and balls, presented in Section 5 of \cite{KMSTTZ03}, the
size of every bucket is $O(\log^c n)$, where $c>0$ is a constant, and no bucket becomes empty w.h.p. We consider epochs
of size $O(\log{n})$, with respect to update operations. During an
epoch, according to Theorem \ref{thm:insdel:logn}, the number of violating points
is expected to be $O(1)$ w.h.p. The extra PST stores
exactly those $O(1)$ violating points. When a new epoch starts, we take
all points from the extra PST and insert them in the respective
buckets in time $O(\log{\log{n}})$ expected w.h.p. Then we need to incrementally update the PST of the upper
level. This is done during the new epoch that just started. In this
way, we keep the PST of the upper level updated and the size of the
extra PST constant. As a result, the update operations are carried out
in $O(\log{\log{n}})$ time expected w.h.p., since
the update of the upper level costs $O(1)$ time w.c.

The 3-sided query can be carried out in the standard way. Assume the
query $[a,b]\times (-\infty,c]$. First we search down the PST of the
upper level for $a$ and $b$. Let $P_a$ be the search path for $a$
and $P_b$ for $b$ respectively. Let $P_m=P_a\cap P_b$. Then, we check
whether the points in the nodes on $P_a\cup P_b$ belong to the answer
by checking their $x$-coordinate as well as their
$y$-coordinate. Then, we check all right children of $P_a - P_m$ as
well as all left children of $P_b - P_m$. In this case we just check
their $y$-coordinate since we know that their $x$-coordinate
belongs in $[a,b]$. When a point belongs in the query, we also
check its two children and we do this recursively. After finishing
with the upper level we go to the respective buckets by following a
single pointer from the nodes of the upper level PST of which the
points belong in the answer. Then we traverse in the same way the
buckets and find the set of points to report. Finally, we check the
extra PST for reported points. In total the query time is
$O(\log{n}+t)$ w.c.

Note that deletions of points do not affect the correctness of the query algorithm. If a non violating point is deleted, it should
reside on the lower level and thus it would be deleted online. Otherwise, the auxiliary PST contains it and thus the deletion is
online again. No deleted violating point is incorporated into the upper level, since by the end of the epoch the PST contains only
inserted violating points.

\begin{theorem}
There exists a dynamic main memory data structure that supports
3-sided queries in $O(\log n + t)$ w.c. time, can be updated in
$O(\log \log n)$ expected w.h.p. and consumes linear space, under the
assumption that the $x$ and $y$-coordinates are continuously drawn
from a $\mu$-random distribution.
\end{theorem}

If we implement the above solution by using EPSTs \cite{ASV99},
instead of PSTs, then the solution becomes I/O-efficient, however the update cost is amortized instead of worst case. Thus we get
that:

\begin{theorem}
There exists a dynamic external memory data structure that supports
3-sided queries in $O(\log_B n + t/B)$ w.c. time, can be updated in
$O(\log_B \log n)$ amortized expected w.h.p. and consumes linear space, under the
assumption that the $x$ and $y$-coordinates are continuously drawn
from a $\mu$-random distribution.
\end{theorem}

\section{The Second Solution for the Smooth and Random Distributions}
\label{second}

We will present the invented data structures in RAM and I/O model respectively.

\subsection{The Second Solution in RAM model}

Our internal memory construction for storing $n$ points in the plane consists of an IS-tree storing the points in sorted order with respect to the $x$-coordinates. On the sorted points, we maintain a weight balanced exponential search tree $T$ with $c_2 = 3/2$ and $c_1 = 2^6$. Thus its height is $\Theta (\log \log n)$. In order to use $T$ as a priority search tree, we augment it as follows. The root stores the point with overall minimum $y$-coordinate.
Points are assigned to nodes in a top-down manner, such that a node $u$ stores the point with minimum $y$-coordinate among the points in $T_u$ that is not already stored at an ancestor of $u$. Note that the point from a leaf of $T$ can only be stored at an ancestor of the leaf and that the $y$-coordinates of the points stored at a leaf-to-root path are monotonically decreasing (\emph{Min-Heap Property}). Finally, every node contains an RMQ-structure on the $y$-coordinates of the points in the children nodes and an array with pointers to the children nodes. Every point in a leaf can occur at most once in an internal node $u$ and the RMQ-structure of $u$'s parent. Since the space of the IS-tree is linear~\cite{MT93,KMSTTZ06}, so is the total space.


\subsubsection{Querying the Data Structure:}

\label{queryds}
Before we describe the query algorithm of the data structure, we will describe the query algorithm that finds all points with $y$-coordinate less than $c$ in a subtree~$T_u$. Let the query begin at an internal node $u$. At first we check if the $y$-coordinate of the point stored at $u$ is smaller or equal to $c$ (we call it a {\em member} of the query). If not we stop. Else, we identify the $t_u$ children of $u$ storing points with $y$-coordinate less than or equal to $c$, using the RMQ-structure of $u$. That is, we first query the whole array and then recurse on the two parts of the array partitioned by the index of the returned point. The recursion ends when the point found has $y$-coordinate larger than $c$ ({\em non-member} point).
%

\begin{lemma}
\label{query}
For an internal node $u$ and value $c$, all points stored in $T_u$ with $y$-coordinate $\leq$$c$ can be found in $O(t + 1)$ time, when $t$ points are reported.
\end{lemma}
\begin{proof}
Querying the RMQ-structure at a node $v$ that contains $t_v$ member points will return at most $t_v + 1$ non-member points. We only query the RMQ-structure of a node $v$ if we have already reported its point as a member point. Summing over all visited nodes we get a total cost of $O\left(\sum_v (2t_v + 1) \right)$$=$$O(t + 1)$.
\qed
\end{proof}

In order to query the whole structure, we first process a 3-sided query $[a,b]\times(- \infty, c]$ by searching for $a$ and $b$ in the IS-tree. The two accessed leaves $a,b$ of the IS-tree comprise leaves of $T$ as well. We traverse $T$ from $a$ and $b$ to the root. Let $P_a$ (resp. $P_b$) be the root-to-leaf path for $a$ (resp. $b$) in $T$ and let $P_m=P_a\cap P_b$. During the traversal we also record the index of the traversed child. When we traverse a node $u$ on the path $P_a - P_m$ (resp. $P_b - P_m$), the recorded index comprises the leftmost (resp. rightmost) margin of a query to the RMQ-structure of $u$. Thus all accessed children by the RMQ-query will be completely contained in the query's $x$-range $[a,b]$. Moreover, by Lem.~\ref{query} the RMQ-structure returns all member points in~$T_u$.

For the lowest node in $P_m$, i.e. the lowest common ancestor (LCA) of $a$ and $b$, we query the RMQ-structure for all subtrees contained completely within $a$ and~$b$. We don't execute RMQ-queries on the rest of the nodes of $P_m$, since they root subtrees that overlap the query's $x$-range. Instead, we merely check if the $x$- and $y$-coordinates of their stored point lies within the query. Since the paths $P_m$, $P_a - P_m$ and $P_b - P_m$ have length $O(\log \log n)$, the query time of $T$ becomes $O(\log \log n + t)$. When the $x$-coordinates are smoothly distributed, the query to the IS-Tree takes $O(\log \log n)$ expected time with high probability~\cite{MT93}. Hence the total query time is $O(\log \log n + t)$ expected with high probability.
\subsubsection{Inserting and Deleting Points:}

\label{insdel}
Before we describe the update algorithm of the data structure, we will first prove some properties of updating the points in $T$. Suppose that we decrease the $y$-value of a point $p_u$ at node $u$ to the value $y'$. Let $v$ be the ancestor node of $u$ highest in the tree with $y$-coordinate bigger than $y'$. We remove $p_u$ from $u$. This creates an ``empty slot'' that has to be filled by the point of $u$'s child with smallest $y$-coordinate. The same procedure has to be applied to the affected child, thus causing a \emph{``bubble down''} of the empty slot until a node is reached with no points at its children. Next we replace $v$'s point $p_v$ with $p_u$ (\emph{swap}). We find the child of $v$ that contains the leaf corresponding to $p_v$ and swap its point with $p_v$. The procedure recurses on this child until an empty slot is found to place the last swapped out point (\emph{``swap down''}). In case of increasing the $y$-value of a node the update to $T$ is the same, except that $p_u$ is now inserted at a node along the path from $u$ to the leaf corresponding to $p_u$.

For every swap we will have to rebuild the RMQ-structures of the parents of the involved nodes, since the RMQ-structures are static data structures. This has a linear cost to the size of the RMQ-structure (Sect.~\ref{prel}).
%
%

\begin{lemma}
\label{dom}
Let $i$ be the highest level where the point has been affected by an update. Rebuilding the RMQ-structures due to the update takes $O(w_i^{c_2 -1})$ time.
\end{lemma}
\begin{proof}
The executed ``bubble down'' and ``swap down'', along with the search for~$v$, traverse at most two paths in $T$. We have to rebuild all the RMQ-structures that lie on the two $v$-to-leaf paths, as well as that of the parent of the top-most node of the two paths. The RMQ-structure of a node at level $j$ is proportional to its degree, namely $O\left( w_j /w_{j-1} \right)$. Thus, the total time becomes $O$$\left(\smash\sum_{j=1}^{i+1}  w_j/w_{j-1} \right)= O$$\left(\smash \sum_{j=0}^{i} w_j^{c_2 -1 } \right) = O$$\left( w_i^{c_2 -1}\right)$.
\qed
\end{proof}

To insert a point $p$, we first insert it in the IS-tree. This creates a new leaf in $T$, which might cause several of its ancestors to overflow. We split them as described in Sec.~\ref{prel}. For every split a new node is created that contains no point. This empty slot is filled by ``bubbling down'' as described above. Next, we search on the path to the root for the node that $p$ should reside according to the Min-Heap Property and execute a ``swap down'', as described above. Finally, all affected RMQ-structures are rebuilt.

To delete point $p$, we first locate it in the IS-tree, which points out the corresponding leaf in $T$. By traversing the leaf-to-root path in $T$, we find the node in $T$ that stores $p$. We delete the point from the node and ``bubble down'' the empty slot, as described above. Finally, we delete the leaf from $T$ and rebalance~$T$ if required. Merging two nodes requires one point to be ``swapped down'' through the tree. In case of a share, we additionally ``bubble down'' the new empty slot. Finally we rebuild all affected RMQ-structures and update the IS-tree.

\noindent
\textbf{Analysis:} We assume that the point to be deleted is selected uniformly at random among the points stored in the data structure. Moreover, we assume that the inserted points have their $x$-coordinates drawn independently at random from an $(n^\alpha, n^{1/2})$-smooth distribution for a constant $1/2$$<$$\alpha $$< $$1$, and that the $y$-coordinates are drawn from an arbitrary distribution. Searching and updating the IS-tree needs $O(\log \log n)$ expected with high probability~\cite{MT93,KMSTTZ06},  under the same assumption for the $x$-coordinates.

%

\begin{lemma}
\label{updint}
Starting with an empty weight balanced exponential tree, the amortized time of rebalancing it due to insertions or deletions is $O(1)$.
\end{lemma}
\begin{proof}
A sequence of $n$ updates requires at most $O$$\left( n/w_i\right)$ rebalancings at level $i$ (Lem.~\ref{rebal}). Rebuilding the RMQ-structures after each rebalancing costs $O$$\left(w_i^{c_2 -1}\right)$ time (Lem.~\ref{dom}). Summing over all levels, the total time becomes~$O$$( \sum_{i=1}^{height(T)} \frac{n}{w_i} \cdot w_i^{c_2 -1}) = O$$( n\sum_{i=1}^{height(T)} w_{i}^{c_2 -2}) $$= O(n)$, when $c_2\smash <\smash 2$.
\qed
\end{proof}

\begin{lemma}
\label{insdelint}
The expected amortized time for inserting or deleting a point in a weight balanced exponential tree is $O(1)$.
\end{lemma}
\begin{proof}
The insertion of a point creates a new leaf and thus $T$ may rebalance, which by Lemma~\ref{updint} costs $O(1)$ amortized time.
Note that the shape of $T$ only depends on the sequence of updates and the $x$-coordinates of the points that have been inserted. The shape of $T$ is independent of the $y$-coordinates, but the assignment of points to the nodes of $T$ follows uniquely from the $y$-coordinates, assuming all $y$-coordinates are distinct.
Let $u$ be the ancestor at level $i$ of the leaf for the new point $p$. For any integer $k\geq 1$, the probability of $p$ being inserted at~$u$ or an ancestor of $u$ can be bounded by the probability that a point from a leaf of~$T_u$ is stored at the root down to the $k$-th ancestor of $u$ plus the probability that the $y$-coordinate of $p$ is among the $k$ smallest $y$-coordinates of the leaves of $T$.
The first probability is bounded by $\sum_{j=i+k}^{height(T)} \frac {2w_{j-1}}{\frac{1}{2} w_j}$, whereas the second probability is bounded by $k\big/\frac{1}{2}w_i$. It follows that $p$ ends up at the $i$-th ancestor or higher with probability at most $O$$\Big( \sum_{j=i+k}^{height(T)} \frac {2w_{j-1}}{\frac{1}{2} w_j} + \frac{k}{\frac{1}{2} w_i} \Big) =
O$$\left( \sum_{j=i+k}^{height(T)} w_{j-1}^{1-c_2} + \frac{k}{w_i} \right) =
O$$\left( w_{i+k-1}^{1-c_2} + \frac{k}{w_i} \right) =
O$$\Big( w_{i}^{ (1-c_2) \smash c_2^{k-1}} + \frac{k}{w_i} \Big) =
O$$\left( \frac{1}{w_i} \right)$ for $c_2 = 3/2$ and $k=3$. Thus the expected cost of ``swapping down'' $p$ becomes $O$$\left( \sum_{i=1}^{height(T)} \frac{1}{w_i}\cdot \frac{w_{i+1}}{w_i} \right) =
O$$\left(\sum_{i=1}^{height(T)} w_i^{c_2 -2}\right) =
O$$\left(\sum_{i=1}^{height(T)} c_1^{ (c_2 -2) c_2^i}\right) = O(1)$ for $c_2<2$.

A deletion results in ``bubbling down'' an empty slot, whose cost depends on the level of the node that contains it. Since the point to be deleted is selected uniformly at random and there are $O\left(n/w_i\right)$ points at level $i$, the probability that the deleted point is at level $i$ is $O\left( 1/w_i\right)$. Since the cost of an update at level $i$ is $O\left( w_{i+1}/w_i\right)$, we get that the expected ``bubble down'' cost is $O\left( \sum_{i=1}^{height(T)} \frac{1}{w_i} \cdot \frac{w_{i+1}}{w_{i}}\right) = O(1)$ for $c_2<2$.
\qed
\end{proof}
\begin{theorem}
In the RAM model, using $O(n)$ space, 3-sided queries can be supported in $O(\log \log n + t/B)$ expected time with high probability, and updates in $O(\log \log n)$ time expected amortized, given that the $x$-coordinates of the inserted points are drawn from an $(n^\alpha,n^{1/2})$-smooth distribution for constant $1/2$$<$$\alpha$$<$$1$, the $y$-coordinates from an arbitrary distribution, and that the deleted points are drawn uniformly at random among the stored points.
\end{theorem}

\subsection{The Second Solution in I/O model}

We now convert our internal memory into a solution for the I/O model. First we substitute the IS-tree with its variant in the I/O model, the ISB-Tree~\cite{KMMSTTZ05}. We implement every consecutive $\Theta(B^2)$ leaves of the ISB-Tree with the data structure of Arge et al.~\cite{ASV99}. Each such structure constitutes a leaf of a weight balanced exponential tree $T$ that we build on top of the $O(n/B^2)$ leaves.

In $T$ every node now stores $B$ points sorted by $y$-coordinate, such that the maximum $y$-coordinate of the points in a node is smaller than all the $y$-coordinates of the points of its children (Min-Heap Property). The $B$ points with overall smallest $y$-coordinates are stored at the root. At a node $u$ we store the $B$ points from the leaves of $T_u$ with smallest $y$-coordinates that are not stored at an ancestor of $u$. At the leaves we consider the $B$ points with smallest $y$-coordinate among the remaining points in the leaf to comprise this list. Moreover, we define the weight parameter of a node at level $i$ to be $w_i$$=$$B^{\smash 2\smash \cdot{\smash (\smash 7\smash /\smash 6\smash )^{\smash i}}}$. Thus we get~$w_{i+1}\smash  =\smash  w_{i}^{\smash 7\smash /\smash 6}$, which yields a height of $\Theta(\log  \log_B n)$. Let $d_i\smash  =\frac{w_i}{w_{i-1}}\smash  =\smash  w_i^{\smash 1\smash /\smash 7}$ denote the {\em degree parameter} for level $i$. All nodes at level $i$ have degree $O(d_i)$. Also every node stores an array that indexes the children according to their $x$-order.

We furthermore need a structure to identify the children with respect to their $y$-coordinates. We replace the RMQ-structure of the internal memory solution with a table. For every possible interval $[k,l]$ over the children of the node, we store in an entry of the table the points of the children that belong to this interval, sorted by $y$-coordinate. Since every node at level $i$ has degree $O(d_i)$, there are $O(d_i^2)$ different intervals and for each interval we store $O(B \cdot d_i)$ points. Thus, the total size of this table is $O(B\cdot d_i^3)$ points or $O(d_i^3)$ disk blocks.

The ISB-Tree consumes $O(n/B)$ blocks~\cite{KMMSTTZ05}. Each of the $O(n/B^2)$ leaves of $T$ contains $B^2$ points. Each of the $n/w_i$ nodes at level $i$ contains $B$ points and a table with $O(B\cdot d_i^3)$ points. Thus, the total space is $O$$\Big($$ n$$ +$$\smash \sum_{i=1}^{\smash height(T)} $$n$$\cdot B$$\cdot$$d_i^3/w_i \Big)  =
O$$\Big(\smash  n\smash  +\smash \sum_{i=1}^{height(T)} n$$\cdot$$ B$$\big/ $$\big(\smash B^{2\cdot\smash {\frac{\smash 7}{6}^{\smash  i}}}\big)^{\frac{\smash  4}{7}} \Big) = O(n)$ points, i.e. $O(n/B)$ disk blocks.


\subsubsection{Querying the Data Structure:}
The query is similar to the internal memory construction. First we access the ISB-Tree, spending $O(\log_B \log n)$ expected I/Os with high probability, given that the $x$-coordinates are smoothly distributed~\cite{KMMSTTZ05}. This points out the leaves of $T$ that contain $a,b$. We perform a 3-sided range query at the two leaf structures. Next, we traverse upwards the leaf-to-root path $P_a$ (resp. $P_b$) on $T$, while recording the index $k$ (resp. $l$) of the traversed child in the table. That costs $\Theta(\log \log_B n)$ I/Os. At each node we report the points of the node that belong to the query range. For all nodes on $P_a - P_b$ and $P_b - P_a$ we query as follows: We access the table at the appropriate children range, recorded by the index~$k$ and $l$. These ranges are always $[k+1,$last child$]$ and $[0,l-1]$ for the node that lie on $P_a - P_b$ and $P_b - P_a$, respectively. The only node where we access a range $[k+1,l-1]$ is the LCA of the leaves that contain $a$ and $b$. The recorded indices facilitate access to these entries in $O(1)$ I/Os. We scan the list of points sorted by $y$-coordinate, until we reach a point with $y$-coordinate bigger than $c$. All scanned points are reported. If the scan has reported all $B$ elements of a child node, the query proceeds recursively to that child, since more member points may lie in its subtree. Note that for these recursive calls, we do not need to access the $B$ points of a node $v$, since we accessed them in $v$'s parent table. The table entries they access contain the complete range of children. If the recursion accesses a leaf, we execute a 3-sided query on it, with respect to $a$ and $b$~\cite{ASV99}.

The list of $B$ points in every node can be accessed in $O(1)$ I/Os. The construction of~\cite{ASV99} allows us to load the $B$ points with minimum $y$-coordinate in a leaf also in $O(1)$ I/Os. Thus, traversing $P_a$ and $P_b$ costs $\Theta(\log \log_B n)$ I/Os worst case. There are $O(\log \log_B n)$ nodes $u$ on $P_a - P_m$ and $P_b - P_m$. The algorithm recurses on nodes that lie within the $x$-range. Since the table entries that we scan are sorted by $y$-coordinate, we access only points that belong to the answer. Thus, we can charge the scanning I/Os to the output. The algorithm recurses on all children nodes whose $B$ points have been reported. The I/Os to access these children can be charged to their points reported by their parents, thus to the output. That allows us to access the child even if it contains only $o(B)$ member points to be reported. The same property holds also for the access to the leaves. Thus we can perform a query on a leaf in $O(t/B)$ I/Os. Summing up, the worst case query complexity of querying $T$ is $O( \log \log_B n + \frac{t}{B})$ I/Os. Hence in total the query costs  $O(\log \log_B n + \frac{t}{B})$ expected I/Os with high probability.
\subsubsection{Inserting and Deleting Points:}
Insertions and deletions of points are in accordance with the internal solution. For the case of insertions, first we update the ISB-tree. This creates a new leaf in the ISB-tree that we also insert at the appropriate leaf of $T$ in $O(1)$ I/Os~\cite{ASV99}. This might cause some ancestors of the leaves to overflow. We split these nodes, as in the internal memory solution. For every split $B$ empty slots ``bubble down''. Next, we update~$T$ with the new point. For the inserted point $p$ we locate the highest ancestor node that contains a point with $y$-coordinate larger than $p$'s. We insert $p$ in the list of the node. This causes an excess point, namely the one with maximum $y$-coordinate among the $B$ points stored in the node, to ``swap down'' towards the leaves. Next, we scan all affected tables to replace a single point with a new one.

In case of deletions, we search the ISB-tree for the deleted point, which points out the appropriate leaf of $T$. By traversing the leaf-to-root path and loading the list of $B$ point, we find the point to be deleted. We remove the point from the list, which creates an empty slot that ``bubbles down'' $T$ towards the leaves. Next we rebalance $T$ as in the internal solution. For every merge we need to ``swap down'' the $B$ largest excess points. For a share, we need to ``bubble down'' $B$ empty slots. Next, we rebuild all affected tables and update the ISB-tree.

\noindent
\textbf{Analysis:} Searching and updating the ISB-tree requires $O(\log_B \log n)$ expected I/Os with high probability, given that the $x$-coordinates are drawn from an $(n/(\log \log n)^{1+\varepsilon}, n^{1/B})$-smooth distribution, for constant~$\varepsilon$$>$$0$~\cite{KMMSTTZ05}.

\begin{lemma}
\label{domext}
For every path corresponding to a ``swap down'' or a ``bubble down'' starting at level $i$, the cost of rebuilding the tables of the paths is $O$$\left( d_{i+1}^3 \right)$ I/Os.
\end{lemma}

\begin{proof}
Analogously to Lem.~\ref{dom}, a ``swap down'' or a ``bubble down'' traverse at most two paths in $T$. A table at level $j$ costs $O$$(d_{j}^3)$ I/Os to be rebuilt, thus all tables on the paths need  $O$$\left(\smash\sum_{j=1}^{i+1}\smash d_j^{\smash 3} \right)= O$$\left( d_{i+1}^{\smash 3} \right)$ I/Os.\qed
\end{proof}

\begin{lemma}
\label{updext}
Starting with an empty external weight balanced exponential tree, the amortized I/Os for rebalancing it due to insertions or deletions is $O(1)$.
\end{lemma}

\begin{proof}
We follow the proof of Lem.~\ref{updint}. Rebalancing a node at level $i$ requires $O$$\left( d_{i+1}^3\smash +\smash B\smash \cdot \smash d_{i}^3\right)$ I/Os (Lem.~\ref{domext}), since we get $B$ ``swap downs'' and ``bubble downs'' emanating from the node. The total I/O cost for a sequence of $n$ updates is $O$$ \big( \smash \sum_{i=1}^{\smash  height(T)} \frac{n}{w_i} \smash \cdot \smash  ( d_{i+1}^3\smash +\smash B\smash \cdot \smash d_{i}^3) \big) \smash =
\smash O$$\big( \smash  n \smash  \cdot \smash  \sum_{i=1}^{height(T)} w_{i}^{\smash -\smash  1 \smash /\smash  2}\smash  + \smash  B \smash \cdot \smash  w_i^{\smash -\smash  4 \smash /\smash  7}\big)\smash =\smash  O \smash  (\smash  n \smash )$.\qed
\end{proof}

\begin{lemma}
\label{insdelext}
The expected amortized I/Os for inserting or deleting a point in an external weight balanced exponential tree is $O(1)$.
\end{lemma}

\begin{proof}
By similar arguments as in Lem.~\ref{insdelint} and considering that a node contains~$B$ points, we bound the probability that point $p$ ends up at the $i$-th ancestor or higher by $O$$\left( B/w_i \right)$. An update at level $i$ costs $O\smash (d_{i+1}^3)\smash =\smash O$$\left(w_i^{\smash 1\smash /\smash 2} \right)$ I/Os. Thus ``swapping down'' $p$ costs $O$$\big( \smash \sum_{i=1}^{height(T)} w_i^{\smash  1 \smash /\smash  2} \smash \cdot\frac{B}{w_i} \big) \smash  = \smash O \smash (\smash 1\smash )$  expected I/Os.
The same bound holds for deleting $p$, following similar arguments as in Lem.~\ref{insdelint}.
\qed
\end{proof}


\begin{theorem}
In the I/O model, using $O(n/B)$ disk blocks, 3-sided queries can be supported in $O(\log \log_B n + t/B)$ expected I/Os with high probability, and updates in $O(\log_B \log n)$ I/Os expected amortized, given that the $x$-coordinates of the inserted points are drawn from an $(n/(\log \log n)^{1+\varepsilon}, n^{1/B})$-smooth distribution for a constant $\varepsilon > 0$, the $y$-coordinates from an arbitrary distribution, and that the deleted points are drawn uniformly at random among the stored points.
\end{theorem}

\section{The Third Solution for the Smooth and the Restricted Distributions}
\label{third}
We would like to improve the query time and simultaneously preserve the update time. For this purpose we will incorporate
to the structure the MPST, which is a static data structure. We
will dynamize it by using the technique of global rebuilding \cite{LO88}, which
unfortunately costs $O(n)$ time.

In order to retain the update time in
the same sublogarithmic levels, we must ensure that at most a
logarithmic number of lower level structures will be violated in a
broader epoch of $O(n)$ updates. Since the violations concern the
$y$-coordinate we will restrict their distribution to the more restricted class,
since Theorem \ref{thm:insdel:n} ensures exactly this property. Thus, the auxiliary PST consumes
at most $O(\log n)$ space during an epoch.

Moreover, we must waive the previous assumption on the $x$-coordinate distribution, as well. Since the query time
of the previous solution was $O(\log n)$ we could afford to pay as much time in order to locate the leaves containing $a$ and $b$.
In this case, though, this blows up our complexity. If, however, we assume that the $x$-coordinates are drawn
from a $(n^{\alpha}, n^{\beta})$-smooth distribution, we can use an IS-tree to index them, given that $0< \alpha , \beta < 1$.
By doing that, we pay w.h.p. $O(\log \log n)$ time to locate $a$ and $b$.

\begin{figure}
\label{fig:2nd:intrnal}
\includegraphics[scale=0.70]{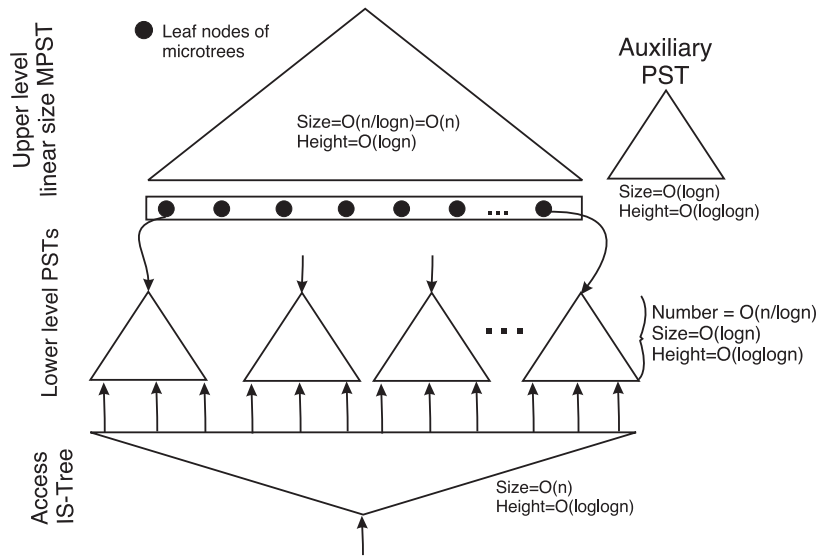}
\caption{The internal memory construction for the restricted distributions}
\end{figure}

When a new epoch starts we take all points from the extra PST and insert them in the respective
buckets in time $O(\log {\log{n}})$ w.h.p. During the epoch we gather all the violating points that should access
the MPST and the points that belong to it and build in parallel a new MPST on them. At the end of the $O(n)$ epoch,
we have built the updated version of the MPST, which we use for the next epoch that just started.
By this way, we keep the MPST of the upper level updated and the size of the extra PST logarithmic.
By incrementally constructing the new MPST we spend $O(1)$ time worst case for each update of the epoch.
As a result, the update operation is carried out in $O(\log \log{n})$ time expected with high probability.

For the 3-sided query $[a,b]\times (-\infty,c]$, we first access the leaves of the lower level that contain $a$ and $b$,
through the IS-tree. This costs $O(\log \log n)$ time w.h.p. Then the query proceeds bottom up in the standard way.
First it traverses the buckets that contain $a$ and $b$ and then it accesses the MPST from the leaves of the buckets' representatives.
Once the query reaches the node of the MPST with $y$-coordinate bigger than $c$, it continues top down to the respective buckets,
which contain part of the answer, by following a single pointer from the nodes of the upper level MPST.
Then we traverse top down these buckets and complete the set of points to report. Finally, we check the auxiliary PST for reported points.
The traversal of the MPST is charged on the size of the answer $O(t)$ and the traversal of the lower level costs $O(\log \log n)$
expected with high probability. Due to Theorem \ref{thm:insdel:n}, the size of the auxiliary PST is with high probability $O(\log n)$,
thus the query spends $O(\log \log n)$ expected with high probability for it. Hence, in total the query time is $O(\log \log n + t)$.

\begin{theorem}
There exists a dynamic main memory data structure that supports
3-sided queries in $O(\log \log n + t)$ time expected w.h.p., can be updated in $O(\log \log n)$ expected w.h.p. and
consumes linear space, under the assumption that the $x$-coordinates
are continuously drawn from a $\mu$-random distribution and the
$y$-coordinates are drawn from the restricted class of distributions.
\end{theorem}

In order to extend the above structure to work in external memory we
will follow a similar scheme with the above structure. We use an auxiliary EPST and index the leaves of the main structure with
and ISB-tree. This imposes that the $x$-coordinates are drawn from a $(n/(\log \log n)^{1+\epsilon}, n^{1-\delta})$-smooth distribution,
where $\epsilon > 0$ and $\delta = 1 -\frac{1}{B}$, otherwise the search bound would not be expected to be doubly logarithmic.
Moreover, the main structure consists of three levels, instead of two. That is, we divide the $n$ elements into
$n' = \frac{n}{\log n}$ buckets of size $\log n$, which we implement as EPSTs (instead of PSTs).
This will constitute the lower level of the whole structure. The $n'$ representatives of the EPSTs are again divided
into buckets of size $O(B)$, which constitute the middle level. The
$n'' = \frac{n'}{B}$ representatives are stored in the leaves of an
external MPST (EMPST), which constitutes the upper level of the whole
structure. 
In total, the space of the aforementioned structures is $O(n' + n'' + n'' \log^{(k)} n'') =
O(\frac{n}{\log n} + \frac{n}{B \log n} + \frac{n }{B \log n} B) =
O(\frac{n}{\log n}) = O(\frac{n}{B})$, where $k$ is such that $\log^{(k)}n'' = O(B)$ holds.

\begin{figure}
\label{fig:2nd:extrnal}
\includegraphics[scale=0.70]{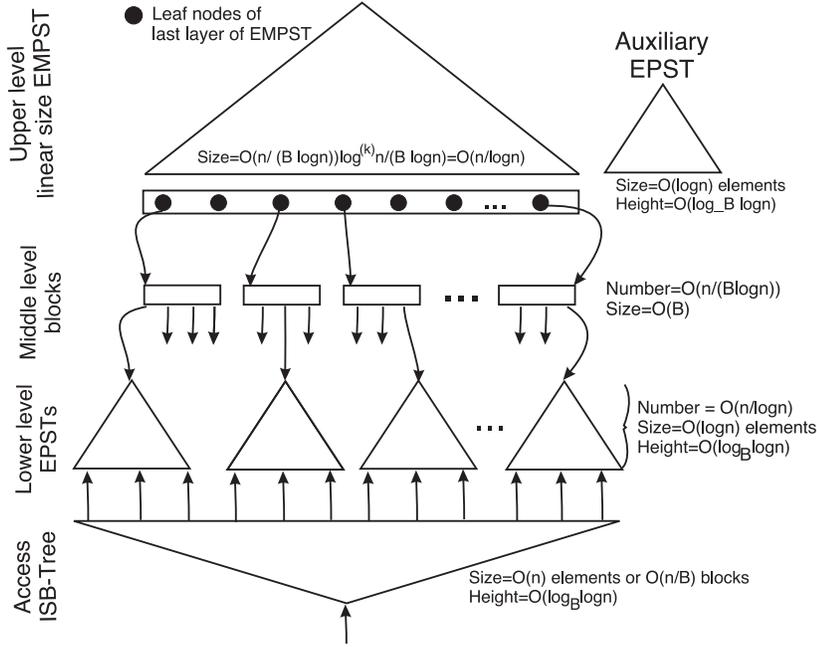}
\caption{The external memory construction for the restricted distributions}
\end{figure}

The update algorithm is similar to the variant of internal memory. The query
algorithm first proceeds bottom up. We locate the appropriate
structures of the lower level in $O(\log_B \log n)$ I/Os w.h.p., due to the assumption on
the $x$-coordinates. The details for this procedure in the I/O model can be found in \cite{KMMSTTZ05}.
Note that if we assume that the $x$-coordinates are drawn from the \emph{grid distribution} with parameters $[1,M]$,
then this access step can be realized in $O(1)$ I/Os. That is done by using an array $A$ of size $M$ as the access data structure.
The position $A[i]$ keeps a pointer to the leaf with $x$-coordinate not bigger than $i$ \cite{SMKLTTV04}.
Then, by executing the query algorithm, we
locate the at most two structures of the middle level that contain the
representative leaves of the EPSTs we have accessed. Similarly we find
the representatives of the middle level structures in the EMPST.
Once we reached the node whose minimum $y$-coordinate is bigger than $c$, the algorithm continues top down.
It traverses the EMPST and accesses the structures of the middle and the lower level that contain parts of the answer.
The query time spent on the EMPST is $O(t/B)$ I/Os. All accessed middle level
structures cost $O(2 + t/B)$ I/Os. The access on the lower level costs $O(\log_B \log n + t/B)$ I/Os. Hence,
the total query time becomes $O(\log_B \log n + t/B)$ I/Os expected with high probability. We get that:

\begin{theorem}
There exists a dynamic external memory data structure that supports
3-sided queries in $O(\log_B \log n + t/B)$ expected w.h.p., can be updated in $O(\log_B \log n)$ expected w.h.p. and
consumes $O(n/B)$ space, under the assumption that the $x$-coordinates
are continuously drawn from a smooth-distribution and the
$y$-coordinates are drawn from the restricted class of distributions.
\end{theorem}

\section{Conclusions}
\label{concl}
We considered the problem of answering three sided range queries of the form $[a, b]\times (-\infty, c]$ under sequences of inserts
and deletes of points, trying to attain linear space and doubly logarithmic expected w.h.p. operation complexities, under
assumptions on the input distributions. We proposed three solutions, which we modified appropriately in order to work for the RAM and
the I/O model. All of them consist of combinations of known data structures that support the 3-sided query operation.

The internal variant of the first solution combines Priority Search Trees \cite{MC81} and achieves $O(\log \log n)$ expected w.h.p.
update time and $O(\log n + t)$ w.c. query time, using linear space. Analogously, the external variant of the first solution
combines External Priority Search Trees \cite{ASV99} and achieves the update operation in $O(\log_B \log n)$ I/Os expected w.h.p.
and the query operation in $O(\log_B n + t/B)$ I/Os amortized expected w.h.p., using linear space. The bounds are true under the
assumption that the $x$ and $y$-coordinates are drawn continuously from $\mu$-random distributions.

The internal variant of the second solution combines exponential weight balanced trees with RMQ structures and achieves $O(\log \log n + t)$ expected query time with high probability and $O(\log \log n)$ expected amortized update time. Analogously, the external variant of the second solution achieves the update operation in $O(\log_B \log n)$ expected amortized I/Os and the query operation in $O(\log \log_B n + t/B)$ expected I/Os with high probability. The main drawback of this solution appears in the I/O-approach, where the block-size factor $B$
is presented in the second logarithm $(O(\log\log_B n))$.

In order to improve the latter, we proposed a third solution with stronger assumptions on the
coordinate distributions. We restricted the $y$-coordinates to be continuously drawn from a restricted distribution and the $x$-
coordinates to be drawn from $(f(n),$ $g(n))$-smooth distributions, for appropriate functions $f$ and $g$, depending on the model. The
internal variant of this solution can be accessed by a IS-tree \cite{KMSTTZ06}, incorporates the Modified Priority Search Tree
\cite{KMSTTV00} and decreases the query complexity to $O(\log \log n + t)$ expected w.h.p., preserving the update and space
complexity. The external variant combines the External Modified Priority Search Tree, which was presented here, with External
Priority Search Trees and is accessed by an ISB-tree \cite{KMMSTTZ05}. The update time is $O(\log_B \log n)$ I/Os expected w.h.p.,
the query time is $O(\log_B \log n + t/B)$ I/Os and the space is linear.

The proposed solutions are practically implementable. Thus, we leave as a future work an
experimental performance evaluation, in order to prove in practice the improved query performance and scalability of the proposed methods.


\begin{thebibliography}{99}
{\small

\bibitem{A96}
Andersson, A.:
Faster deterministic sorting and searching in linear space.
In: Proc. IEEE FOCS. (1996)  pp. 135-141.

\bibitem{ABR00}
Alstrup, S., Brodal, G.S., Rauhe, T.:
New data structures for orthogonal range searching.
In: IEEE IEEE FOCS (2000), pp. 198–207.

\bibitem{AE99}
Agarwal, P., Erickson, J.:Geometric range rearching and its relatives.
In Chazelle, B., Goodman, J., Pollack, R., eds.: Advances in Discrete
and Computational Geometry. Contemporary Mathematics.
American Mathematical Society Press (1999), pp. 1-56.

\bibitem{AM93}
Andersson, A., Mattsson, C.:
Dynamic interpolation search in o(log log n) time.
In: Proc. ICALP. Volume 700 of Springer LNCS. (1993), pp. 15-27.

\bibitem{ASV99}
Arge, L., Samoladas, V., Vitter, J.S.: On two-dimensional indexability and optimal range search indexing.
In: Proc. ACM SIGMOD-SIGACT-SIGART PODS. (1999), pp. 346-357.

\bibitem{AT07}
Andersson, A., Thorup, M.:
Dynamic ordered sets with exponential search trees.
J. ACM \textbf{54}(3) (2007) ~13.

\bibitem{AV96}
Arge, L., Vitter, J.S.:
Optimal dynamic interval management in external memory (extended abstract).
In: Proc. IEEE FOCS. (1996),  pp. 560-569.

\bibitem{BBD+02}
Babcock, B., Babu, S., Datar, M., Motwani, R., Widom, J.:
Models and Issues in Data Stream Systems,
{\it Proceedings of PODS} (2002), pp.1-16.

\bibitem{BG90}
Blankenagel, G., Gueting, R.H.:
XP-trees-External priority search trees,
Technical report, FernUniversitB Hagen, Informatik-Bericht, Nr.92, 1990.

\bibitem{BGOSW96}
B.~Becker, S. Gschwind, T. Ohler, B. Seeger and P.~Widmayer,
``An asymptotically optimal multiversion B-tree'',
{\it The VLDB Journal}, pp.264-275, 1996.

\bibitem{CXPS+04}
Cheng, R., Xia, Y., Prabhakar, S., Shah, R., Vitter, J.S.:
Efficient Indexing Methods for Probabilistic Threshold Queries over Uncertain Data,
{\it Proceedings of VLDB} (2004), pp.876-887.

\bibitem{dKOS98}
deBerg, M., van Kreveld, M., Overmars, M., Schwarzkopf, O.:
{\it Computational Geometry, algorithms and applications},
Springer, 1998.

\bibitem{FH07}
Fischer, J., Heun, V.:
A\emph{ New Succinct Representation of RMQ-Information and Improvements in the Enhanced Suffix Array}, Proceedings of the International Symposium on Combinatorics, Algorithms, Probabilistic and Experimental Methodologies, Lecture Notes in Computer Science, 4614, Springer-Verlag, pp. 459–470.

\bibitem{FMNT87}
O. Fries, K. Mehlhorn, S. Naher and A. Tsakalidis,
``A loglogn data structure for three sided range queries'',
{\it Information Processing Letters}, 25, pp.269-273, 1987.

\bibitem{GT85}
H. N. Gabow and R. E. Tarjan,
``A linear-time algorithm for a special case of disjoint set union'',
{\it Journal of Computer and System Sciences}, 30, pp.209-221, 1985.

\bibitem{G94}
Gusfield, D.:
``{\it Algorithms on Strings, Trees and Sequences, Computer Science and Computational Biology}'', Cambridge University Press,1994.

\bibitem{GG98}
Gaede, V., Gfinther, O.:
Multidimensional access methods,	
{\it ACM Computing Surveys}, 30(2) (1998) pp.170-231.

\bibitem{HT84}
Harel, D., Tarjan, R.E.:
Fast algorithms for finding nearest common ancestors.
SIAM J. Comput. \textbf{13}(2) (1984), pp. 338-355.

\bibitem{IKO87}
Icking, C., Klein, R., Ottmann, T.:
Priority search trees in secondary memory,
{\it Proceedings of Graph-Theoretic Concepts in Computer Science}, LNCS 314 (1987), pp.84-93.

\bibitem{K77}
D.E.~Knuth,
``Deletions that preserve randomness'',
{\it IEEE Transactions on Software Engineering}, 3, pp.351-359, 1977.

\bibitem{KMMSTTZ05}
Kaporis, A.C., Makris, C., Mavritsakis, G., Sioutas, S., Tsakalidis, A.K., Tsichlas, K., Zaroliagis, C.D.:
ISB-tree: A new indexing scheme with efficient expected behaviour.
In: Proc. ISAAC. Volume 3827 of Springer LNCS. (2005), pp. 318-327.

\bibitem{KMSTTV00}
N. Kitsios, C. Makris, S. Sioutas, A. Tsakalidis, J. Tsaknakis, B. Vassiliadis,
``2-D Spatial Indexing Scheme in Optimal Time'',
{\it Proceedings of ADBIS-DASFAA}, pp.107-116, 2000.

\bibitem{KMSTTZ03}
Kaporis, A., Makris, C., Sioutas, S., Tsakalidis, A., Tsichlas, K., Zaroliagis,C.:
Improved bounds for finger search on a {RAM}.
In: Proc. ESA. Volume 2832 of Springer LNCS. (2003), pp. 325-336.

\bibitem{KMSTTZ06}
Kaporis, A., Makris, C., Sioutas, S., Tsakalidis, A., Tsichlas, K., Zaroliagis,C.:
Dynamic interpolation search revisited.
In: Proc. ICALP. Volume 4051 of Springer LNCS. (2006), pp. 382-394.

\bibitem{KRVV93}
Kanellakis, P.C., Ramaswamy, S., Vengroff, D.E., Vitter, J.S.:
Indexing for data models with constraints and classes.
In: Proc. ACM SIGACT-SIGMOD-SIGART PODS (1993), pp.233-243.

\bibitem{KRVV96}
PKanellakis, P.C., Ramaswamy, S., Vengroff, D.E., Vitter, J.S.:
Indexing for data models with constraints and classes,
{\it Journal of Computer and System Sciences}, 52(3) (1996), pp.589-612.

\bibitem{LO88}
Levcopoulos, C., Overmars, M.H.:
Balanced Search Tree with O(1) Worst-case Update Time,
{\it Acta Informatica}, 26 (1988), pp.269-277.

\bibitem{MBP06}
Mouratidis, K., Bakiras, S., Papadias, D.:
Continuous Monitoring of Top-$k$ Queries over Sliding Windows,
{\it In Proc. of SIGMOD} (2006), pp.635-646.

\bibitem{MC81}
McCreight, E.M.:
Priority search trees.
SIAM J. Comput. \textbf{14}(2) (1985), pp. 257-276.

\bibitem{MT93}
Mehlhorn, K., Tsakalidis, A.:
Dynamic interpolation search.
J. ACM \textbf{40}(3) (1993),  pp. 621-634.

\bibitem{O88}
M. H. Overmars,
``Efficient data structures for range searching on a grid'',
{\it Journal of Algorithms}, 9, pp.254-275, 1988.

\bibitem{RS94}
Ramaswamy, S., Subramanian, S.:
Path caching: A technique for optimal external searching,
{\it Proceedings of ACM PODS} (1994), pp.25-35.

\bibitem{SMKLTTV04}
Sioutas, S., Makris, C., Kitsios, N., Lagogiannis, G., Tsaknakis, J., Tsichlas, K., Vassiliadis, B.:
Geometric Retrieval for Grid Points in the RAM Model.
J. UCS \textbf{10}(9) (2004), pp. 1325-1353.

\bibitem{SR95}
Subramanian, S.,	Ramaswamy, S.:
The P-range tree: A new data structure for range searching in secondary memory,
{\it Proceedings of ACM SODA} (1995), pp.378-387.

\bibitem{T98}
Thorup, M.:
Faster deterministic sorting and priority queues in linear space.
In: Proc. ACM-SIAM SODA. (1998)  pp. 550-555.

\bibitem{V01}
Vitter, J.S.:
External memory algorithms and data structures: dealing with massive data,
{\it ACM Computing Surveys}, 33(2), pp.209-271, 2001.

\bibitem{W92}
Willard, D.E.:
Applications of the fusion tree method to computational geometry and searching,
{\it Proceedings of ACM-SIAM SODA}(1992) pp.286-295.

\bibitem{W00}
Willard, D.E.:
Examining computational geometry, van emde boas trees, and hashing from the perspective of the fusion tree.
SIAM J. Comput. \textbf{29}(3) (2000), pp. 1030-1049.
}

\end{thebibliography}
\end{document}